\documentclass[sn-mathphys,Numbered]{sn-jnl}
\usepackage{graphicx}
\usepackage{multirow}
\usepackage{amsmath,amssymb,amsfonts}
\usepackage{amsthm,bbm}
\usepackage{mathrsfs}
\usepackage[title]{appendix}
\usepackage{xcolor}
\usepackage{textcomp}
\usepackage{manyfoot}
\usepackage{booktabs}
\usepackage{algorithm}
\usepackage{algorithmicx}
\usepackage{algpseudocode}
\usepackage{listings}

\theoremstyle{thmstyleone}
\newtheorem{theorem}{Theorem}
\newtheorem{proposition}[theorem]{Proposition}
\newtheorem{lemma}[theorem]{Lemma}

\theoremstyle{thmstyletwo}

\theoremstyle{thmstylethree}
\newtheorem{definition}{Definition}

\begin{document}

\title{Event-based quantum contextuality theory}

\author[1]{\fnm{Songyi} \sur{Liu}}\email{liusongyi@buaa.edu.cn}

\author*[1]{\fnm{Yongjun} \sur{Wang}}\email{wangyj@buaa.edu.cn}

\author[1]{\fnm{Baoshan} \sur{Wang}}\email{bwang@buaa.edu.cn}

\author[1]{\fnm{Yunyi} \sur{Jia}}\email{by2309005@buaa.edu.cn}

\affil*[1]{\orgdiv{School of Mathematical Sciences}, \orgname{Beihang University}, \orgaddress{ \city{Beijing}, \postcode{100191}, \country{China}}}

\abstract{Fully revealing the mathmatical structure of quantum contextuality is a significant task, while some known contextuality theories are only applicable for rank-1 projectors. That is because they adopt the observable-based definitions. This paper overcomes the challenges faced by some known contextuality theories by establishing an event-based contextuality theory with exclusive partial Boolean algebra, which is used to describe the contextual systems with local consistency and exclusivity principle. Our theory provides a precise mathematical framework for quantum contextuality, which can handle the scenarios composed of general projectors, and introduces a more complete contextuality hierarchy. We conclude that the Kochen-Specker contextuality is equivalent to the state-independent strong contextuality for finite dimensional quantum systems. Therefore, when considering both the strength and proportion of contextual quantum states, Kochen-Specker contextuality is the strongest.}

\keywords{Quantum contextuality, Partial Boolean algebra, Exclusivity graph approach, Sheaf theory approach}

\maketitle
\tableofcontents
\section*{Declarations}

\bmhead{Competing interests}
The authors have no relevant financial or non-financial interests to disclose.

\section{Introduction}

The nonexistence of hidden-variables in quantum mechanics was proved by Bell nonlocality\citep{Bell1964On} and Kochen-Specker (KS) theorem\citep{Kochen1967The} in 1960s. The Bell nonlocality has inspired the research in applications such as quantum computation, quantum communication channel and quantum cryptography\citep{Brunner2014Bell}, while the Kochen-Specker theorem pioneers a field later known as the contextuality theory\citep{Budroni2022Kochen}. In 2008, the Klyachko-Can-Binicioglu-Shumovsky (KCBS) experiment\citep{Alexander2008Simple} verified that the contextuality is more fundamental than the nonlocality as a quantum feature. Quantum contextuality constitutes a critical treasure trove of resources in quantum mechanics\citep{Mark2014Contextuality}. \par

Fully revealing the mathematical structure of quantum contextuality is a tough and significant task\citep{Budroni2022Kochen}. There are two perspectives for it: observable perspective (OP) and effect perspective (EP). OP focuses on the ideal measurements, with the sheaf theory approach\citep{Abramsky2011sheaf} and exclusivity graph approach\citep{Adan2014Graph} as representative examples. EP, such as the compatibility hypergraph\citep{Acin2015A}, focuses on the nonideal measurements. The two perspectives are equivalent because non-ideal measurements can be realized by the ideal measurements on the higher dimensional Hilbert space. We adopt the observable perspective in this paper.\par

To align with physical intuition, the most theories are designed to be observable-based. In simple terms, if a theory defines events with observables rather than defines observables with events, we call it an ``observable-based contextuality theory". For details, an event is usually defined as joint outcomes $(A_1=a_1,A_2=a_2,...,A_n=a_n)$, where $A_1,A_2...,A_n$ are compatible observables.\par

The observable-based contextuality theories have attained great accomplishments. The exclusivity graph is applied to find the minimal vectors set with Kochen-Specker assignments\citep{Xu2020Proof} and the minimal vectors set with state-independent contextuality (SIC)\citep{Cabello2016Quantum}, but these conclusions are only made in the case of rank-1 projectors (the projectors onto one-dimensional spaces). The sheaf theory approach provides a mathematical framework for the hierarchy of quantum contextuality\citep{Abramsky2011sheaf}, which includes probabilistic, logical and strong contextuality to abstract the Bell\citep{Bell1964On}, Hardy\citep{Hardy1993Nonlocality} and GHZ\citep{Greenberger1989Going} (Greenberger, Horne and
Zeilinger)-type proof of KS theorem. However, the hierarchy faces challenges in definition in the case of general scenarios, and the recent nonclassical models, such as the Yu-Oh set\citep{Yu2012State}, have not yet been taken into account\citep{Budroni2022Kochen}.\par

From a mathematical prospective, using the observable-based definitions is risky, because the fundamental objects of probability theories, whether classical or quantum, are the event algebras. The measurement outcomes of observables do not fully capture the structure of events. For example, let $|0\rangle=(0,0,1)$, $|1\rangle=(0,1,0)$, $|2\rangle=(1,0,0)$, $|x\rangle=(1/\sqrt{2},1/\sqrt{2},0)$, $|y\rangle=(1/\sqrt{2},-1/\sqrt{2},0)$, $\hat{A}=a_0|0\rangle\langle0|+a_1|1\rangle\langle1|+a_2|2\rangle\langle2|$, and $\hat{B}=b_0|0\rangle\langle0|+b_1|x\rangle\langle x|+b_2|y\rangle\langle y|$. Then events $(\hat{A}=a_0)$ and $(\hat{B}=b_1)$ are exclusive. However, they are not defined to be exclusive in the observable-based contextuality theories. We believe that it is the main reason why revealing the mathematical structure of quantum contextuality encounters challenges.\par

To provide a precise mathematical framework for quantum contextuality, we use the exclusive partial Boolean algebra to establish an event-based contextuality theory. The partial Boolean algebra ($pBA$) is introduced by Kochen and Specker\citep{Kochen1967The} to present the original proof of KS theorem. In 2015, Kochen initiated the works of reconstructing the foundation of quantum mechanics with $pBA$\citep{Kochen2015Reconstruction}. \par

$pBA$ naturally satisfies the local consistency\citep{Abramsky2015Contextuality,Ramanathan2012Generalized}, also known as non-signaling in Bell experiments\citep{Popescu1994Quantum}, which means that the probability distributions on different contexts induce no contradiction. However, $pBA$ may violate the exclusivity principle\citep{Fritz2013Local,Adan2012Specker,Abramsky2020The}, which means that the probability sum of exclusive events is not more than 1. In 2021, Abramsky and Barbosa generalize the exclusivity principle to the $pBA$, and get the exclusive partial Boolean algebra ($epBA$)\citep{Abramsky2020The}. An $epBA$ is determined by its atom graph in the finite dimensional cases\citep{Liu2025Atom}, which provides a concise and precise approach to describe the contextual systems satisfying local consistency and exclusivity principle. It is worth noting that the observable-based contextuality theories usually do not consider the two principles simultaneously, which leads to the possibility of violating one of the them.\par

In Sect. \ref{sec-obct} of this paper, we introduce the basic contextuality notions and two advanced observable-based contextuality theories: the sheaf theory approach and the exclusivity graph approach, analyzing the challenges faced by them. In Sect. \ref{Event-based contextuality theory}, we develop the work of \cite{Abramsky2020The} and \cite{Liu2025Atom} to establish a unified, event-based mathematical theory for quantum contextuality. The correctness of our theory is guaranteed by the theorem \ref{thm-contextuality}. In Sect. \ref{sec-quantum-contextuality}, we formalize a number of significant notions of quantum contextuality, and depict the major quantum scenarios and models. In Sect. \ref{sec-hierarchy}, we prove that the KS contextuality is equivalent to the state-independent strong contextuality for the finite dimensional quantum systems by the theorem \ref{KS-SISC}, which presents a more complete two-dimensional hierarchy for quantum contextuality. In Sect. \ref{sec-conclusion}, we summarize our work and list some questions for further research.

\section{Observable-based contextuality theories}\label{sec-obct}

\subsection{Fundamental concepts}

If a system satisfies the hidden-variable theory, all the observables can be assigned values simultaneously by the hidden-variables $\lambda$, and the randomness of the system is exactly the randomness of $\lambda$.\par

For details, let $X=\{A_1,...,A_n\}$ be a set of observables, and $M$ a family of subsets of $X$. $C\in M$ is called a context, usually defined to be the jointly commeasurable observables. $O_C$ is the set of joint outcomes of $C$, and $e\in O_C$ is called an event. Let $\{p_C:O_C\to[0,1]\ |\ C\in M\}$ be a set of probability distributions over all the contexts. $\{p_C\}$ is said to be a non-contextual hidden-variable model (NCHV model) if there exist hidden-variables $\lambda\in\Lambda$ such that for any context $C$, $p_C(e)=\sum\limits_{\lambda\in\Lambda}p_{\Lambda}(\lambda)\delta(e|\lambda)$, where $e\in O_C$, $p_{\Lambda}$ is a probability distribution on $\Lambda$ and $\delta(e|\lambda)$ is the deterministic conditional probability of $e$ satisfying the factorizability property\citep{Abramsky2011sheaf}: $\delta(e_1\land e_2|\lambda)=\delta(e_1|\lambda)\delta(e_2|\lambda)$ if $e_1,e_2$ are outcomes of a same context. Furthermore, if there exists a $\{p_C\}$ in a system which is not a NCHV model, then the system admits no hidden-variable, in other words, it witnesses the contextuality.\par

In some versions, the deterministic $\delta(e|\lambda)$ is replaced by the probabilistic $p(e|\lambda)=p((A_i=a_i|A_i\in C)|\lambda)=p(\prod\limits_{A_i\in C}p(A_i=a_i|\lambda))$. The two definitions are equivalent, because the existence of $\delta(e|\lambda)$ is equivalent with the existence of $p(e|\lambda)$\citep{Budroni2022Kochen}. In fact, the hidden-variable $\lambda$ has been proved to be exactly the global value-assignment to the all observables $X$\citep{Abramsky2011sheaf}. Therefore, we adopt the deterministic definition for simplicity.\par

The notions above present the classical contextuality theory, which underlies the marginal problem definition\citep{Fritz2013Entropic}. Different contextuality theories use different terminology and definitions. The terminological multiplicity often confuses readers and researchers in this field\citep{Abramsky2023Quantum}. For example, the observables set $X$ and contexts set $M$ describe the structure of systems. They are called ``marginal scenario" in the marginal problem definition and ``graphical measurement scenario" in the sheaf theory approach. In the exclusivity graph approach, they are described by the exclusivity graphs. For unifying the terminology, we use ``scenario" in this paper.\par

Another difference among contextuality theories is reflected in the constraints for $\{p_C\}$. The classical definition of $\{p_C\}$ imposes no restrictions other than the probability distribution, which is so general that it is not applicable to the research of quantum mechanics. Therefore, most contextuality theories impose constraints known as ``quantum principles" on $\{p_C\}$ to bring it closer to quantum states. Two of the most important principles are the local consistency and exclusivity principle.

\subsection{Local consistency and exclusivity principle}

For the set $\{p_C\}$, if for any two contexts $C$ and $C'$, $p_C|_{C\cap C'}=p_{C'}|_{C\cap C'}$,  then $\{p_C\}$ is said to satisfy local consistency in \cite{Abramsky2015Contextuality}, or non-disturbance in \cite{Ramanathan2012Generalized}. In Bell experiments, the principle is called no-signalling\citep{Popescu1994Quantum}, which means that information cannot be transmitted faster than the speed of light. In general scenarios, it means that the same event in different contexts can only be assigned a consistent probability.\par

If two events cannot occur simultaneously, they are said to be exclusive to each other, which is usually defined by the distinction between outcomes of observables in observable-based contextuality theories. Furthermore, if for any set of mutually exclusive events $E$ we have $\sum\limits_{e\in E}p_{C}(e)\leq 1$, where $e\in O_{C}$, then $\{p_C\}$ is said to satisfy exclusivity principle, also known as local orthogonality in \cite{Fritz2013Local}, Specker's exclusivity principle in \cite{Adan2012Specker}, probabilistic exclusivity principle in \cite{Abramsky2020The} and so on.\par

Local consistency and exclusivity principle are satisfied by quantum mechanics, and they are usually seen as features of quantum states. However, we claim that the two principles are actually the features of the structure of quantum systems. As long as a system possesses specific structure, any probability distributions on it will satisfy the two principles, regardless of how these probability distributions are induced. This conclusion addresses how to establish a ``natural" contextuality theory for quantum mechanics. We will elaborate on these details in section \ref{Event-based contextuality theory}.

The status of local consistency and exclusivity principle varies across different contextuality theories. In the marginal problem definition and sheaf theory approach, a set $\{p_C\}$ satisfying local consistency is referred to as a model, while in the exclusivity graph approach, the models are defined to satisfy exclusivity principle. Therefore, the models in these theories may violate one of the two principles. In the following subsections, we introduce two of these known observable-based contextuality theories, and explain the challenges they face.\par

\subsection{Challenges within the sheaf theory approach}

The sheaf theory approach is introduced by Abramsky and Brandenburger to get a mathematical framework of contextuality theory\citep{Abramsky2011sheaf}, whose definitions have undergone several refinements\citep{Abramsky2017Contextual,Abramsky2020The}. The fundamental concepts are defined below.

\begin{definition}
 Let $X$ be a set of observables. $\frown$ is the commeasurable relation on $X$. $O=\{O_x\}_{x\in X}$ is a family of sets, where $O_x$ is the outcomes of $x$. Then triple $(X, \frown, O)$ is called a graphical measurement scenario. A clique $C\subseteq X$ of relation $\frown$ is called a context. $Kl(\frown)$ denotes the set of contexts.\label{sheaf1}
\end{definition}

\begin{definition}
 Define $\xi(C):=\prod_{x\in C}O_x$. Let $\{p_C:\xi(C)\to[0,1]\ |\ C\in Kl(\frown)\}$ be probability distributions over the contexts. If $\{p_C\}$ satisfies local consistency, that is, for any $C'\subseteq C$, $p_{C'}=p_C|_{C'}$, then call $\{p_C\}$ an empirical model.\label{sheaf2}
\end{definition}

Therefore, the scenarios are depicted by triples $(X, \frown, O)$, contexts set $M=Kl(\frown)$, and models are defined to satisfy local consistency.

The sheaf theory approach has conducted extensive research on contextuality and provided a hierarchy. We only introduce the strong contextuality here, which is used to abstract the GHZ-type proof of KS theorem. The definitions are shown below.

\begin{definition}
An empirical model $\{p_C\}$ is said to be non-contextual if there is a global probability distribution $d$ on $\xi(X)=\prod_{x\in X}O_x$ such that $d|_{C}=p_C$ for any context $C\in Kl(\frown)$,\label{sheaf3}
\end{definition}

Note that \cite{Abramsky2011sheaf} has proved that the existence of hidden-variable $\lambda$ is equivalent with the existence of global probability distribution $d$, so the definition \ref{sheaf3} is consistent with the NCHV models.

\begin{definition}
An empirical model $\{p_C\}$ is said to be strongly contextual if for any global event $s\in\xi(X)$, there exists a context $C\in Kl(\frown)$ such that $p_C(s|_C)=0$.\label{sheaf4}
\end{definition}

The GHZ-state\citep{Greenberger1989Going} and any quantum state on Kochen-Specker set are strongly contextual. We will elaborate on them in section \ref{sec-quantum-contextuality}.

Although the sheaf theory approach provides a precise mathematical framework, the observable-based definitions pose significant risks to it. For example, the definition \ref{sheaf2} of empirical model requires local consistency, but it may be not "really" locally consistent in practical applications. Let $\hat{A}=a_0|0\rangle\langle0|+a_1|1\rangle\langle1|+a_2|2\rangle\langle2|$ and $\hat{B}=b_0|0\rangle\langle0|+b_1|x\rangle\langle x|+b_2|y\rangle\langle y|$ as before. Then $\hat{A}$ and $\hat{B}$ are not commeasurable but $\hat{A}=a_0$ is equivalent to $\hat{B}=b_0$. In the graphical measurement scenario, $X=\{\hat{A},\hat{B}\}$, $Kl(\frown)=\{C_A,C_B\}$ where $C_A=\{\hat{A}\}$ and $C_B=\{\hat{B}\}$. Therefore, any probability distributions $\{p_{C_A},p_{C_B}\}$ are empirical models. However, $\{p_{C_A},p_{C_B}\}$ is not really locally consistent when $p_{C_A}(\hat{A}=a_0)\neq p_{C_B}(\hat{B}=b_0)$.\par

Furthermore, the rigor of definitions of contextuality is undermined. We take the strong contextuality as an example. The CEG (Cabello, Estebaranz and Garc$\acute{\i}$a-Alcaine) set\citep{Cabello1997Bell} is the simplest KS vector set\citep{Xu2020Proof}, consisting of 18 unit vectors in $\mathbbm{R}^4$ and 9 contexts. Each context contains 4 mutually orthogonal vectors. The scenario is shown in the Fig.\ref{CEG}, where the vectors upon vertices are the non-normalized coordinates for corresponding projectors, and a straight line or a circumference represents a context.

\begin{figure}[H]
    \centering
    \includegraphics[width=0.5\linewidth]{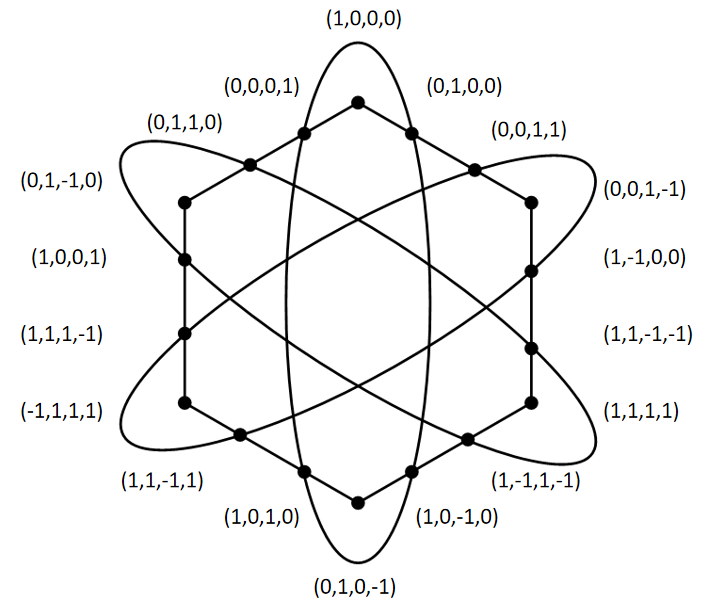}
    \caption{The CEG set}
    \label{CEG}
\end{figure}

Every vector corresponds to a rank-1 projector, or a 0-1 observable. Denote these vectors $v_i$ $(i=1,...,18)$. Then $X=\{|v_i\rangle\langle v_i|\}_{i=1}^{18}$. Each global event $s\in\xi(X)$ make a 0-1 value-assignment to $X$. However, the CEG set is a KS set, which means it's impossible to assign a 0-1 value to each vector, such that exactly one vector is assigned 1 in every context. Therefore, there is at least one context $C'$ such that all $|v\rangle\langle v|\in C'$ are assigned 0 or more than two $|v\rangle\langle v|\in C'$ are assigned 1. On the other hand, each quantum state $\rho$ induces an empirical model $\{\rho_C\}$ such that only assignments 0001, 0010, 0100 and 1000 are possible (that is, $\rho_C>0$) for all contexts $C$. Thus $\rho_{C'}(s|_{C'})=0$. In conclusion, any quantum states $\rho$ are strongly contextual on CEG set.\par

However, one can depict the same scenario in another way. Notice that each context $C_j$ $(j=1,...,9)$ corresponds to a complete observable $\hat{A}_j=\lambda_1|v_{j_1}\rangle\langle v_{j_1}|+\lambda_2|v_{j_2}\rangle\langle v_{j_2}|+\lambda_3|v_{j_3}\rangle\langle v_{j_3}|+\lambda_4|v_{j_4}\rangle\langle v_{j_4}|$. In quantum mechanics, observables $\{\hat{A}_j\}_{j=1}^{9}$ and projectors $\{|v_i\rangle\langle v_i|\}_{i=1}^{18}$ establish the same scenario. Let $X=\{\hat{A}_j\}_{j=1}^{9}$. It is clear that any pairs of these observables are not commeasurable, so there are total 9 contexts $C_1=\{\hat{A}_1\}, ..., C_9=\{\hat{A}_9\}\in Kl(\frown)$. The global event $s=(\hat{A}_1,..., \hat{A}_9=\lambda_1,...,\lambda_1)$ is restricted to $s|_{C_j}=(\hat{A}_j=\lambda_1)$ $(j=1,...,9)$. It is easy to find a quantum state to make all the $s|_{C_j}$ are possible. For instance, the superposition state $\rho$ of (non-normalized) $(1,0,0,0)$, $(0,1,0,0)$, $(1,1,1,1)$, $(1,-1,1,-1)$, $(-1,1,1,1)$ and $(1,1,1,-1)$ (note that these vectors associate to all the 9 contexts $C_j$, so simply let them correspond to the events $(\hat{A}_j=\lambda_1)$, $j=1,...,9$). Therefore, $\rho_{C_j}(s|_{C_j})>0$ for $j=1,...,9$, so $\rho$ is not strong contextual, which derives a contradiction.\par

These problems indicate that the sheaf theory approach for contextuality is not well-defined. In fact, the sheaf theory approach is applicable only when $X$ consists of only rank-1 projectors.\par

\subsection{Challenges within the exclusivity graph approach}\label{subsec-ex-graph}

The exclusivity graph approach\citep{Adan2010Contextuality,Adan2014Graph} is introduced by Cabello et al. The relevant concepts are defined below.\par

\begin{definition}
Let $X=\{A_1,...,A_n\}$ be a set of observables, $C_1,C_2\in M$ two contexts. Then $e_1\in O_{C_1}$ and $e_2\in O_{C_2}$ are said to be exclusive if they differ in an observable $A\in C_1\cap C_2$. \label{exclusivity1}
\end{definition}

\begin{definition}
An exclusivity graph is a simple graph whose vertices set is $\{e\in O_C|C\ is\ maximal\ in\ M\}$, and $e_1,e_2$ are adjoint iff $e_1,e_2$ are exclusive.\label{exclusivity2}
\end{definition}

\begin{definition}
An model $p$ on an exclusivity graph $G$ is defined by a function $p:V(G)\to[0,1]$ such that $\sum\limits_{v\in C}p(v)\leq 1$ for any clique $C\subseteq V(G)$.\label{exclusivity3}
\end{definition}

Therefore, the scenarios are depicted by graphs, contexts are depicted by cliques and models are defined to satisfy exclusivity principle.

The exclusivity graph provides a powerful tool for the research on finite quantum systems, which can be described by graphs. Its definition of contextuality is that

\begin{definition}
An model $p$ on an exclusivity graph $G$ is said to be non-contextual if $p$ is a convex combination of $0-1$ models $\delta:V(G)\to\{0,1\}$.\label{exclusivity4}
\end{definition}

A rough proof was provided in \cite{Adan2010Contextuality} and \cite{Adan2014Graph} that the definition \ref{exclusivity4} is consistent to the NCHV models. In other words, the NCHV models are exactly the convex combination of deterministic models. We will give a precise proof of this conclusion in the theorem \ref{thm-contextuality}.

It should be pointed out that the exclusivity graph approach aimed to define events as an equivalence class, in other words, two events are equivalent if their probability are equal in either state\citep{Adan2014Graph}. However, the definition regularly does not work in theoretical practice, because in the observable-based theory, $X$ does not induce equivalent events like $\hat{A}$ in quantum mechanics without additional conditions. Therefore, the exclusivity graph approach sometimes provides incorrect exclusive relations.\par

For instance, the Klyachko-Can-Binicio$\breve{g}$lu-Shumovsky (KCBS) experiment\citep{Alexander2008Simple} considered five observables $\{\hat{A}_0,...,\hat{A}_4\}$, where $\hat{A}_i=2\hat{S}_{l_i}^2-1$, $\hat{S}_{l_i}$ is the spin projection operators onto directions $l_i$, and $l_i\bot l_{i+1}$ with sum modulo 5. Therefore, if $\hat{A}_i=-1$, then $\hat{A}_{i+1}=1$. In other words, events $(\hat{A}_i=-1)$ and $(\hat{A}_{i+1}=-1)$ are exclusive. But in exclusivity graph approach, they are not considered to be exclusive. In fact, the observables in KCBS experiment form five contexts, $\{\hat{A}_0,\hat{A}_1\}$, ..., $\{\hat{A}_4,\hat{A}_0\}$. Each context has 4 events, and the exclusivity graph of all the events is the Fig. \ref{2014}\citep{Adan2014Graph}.

\begin{figure}[H]
    \centering
    \includegraphics[width=0.5\linewidth]{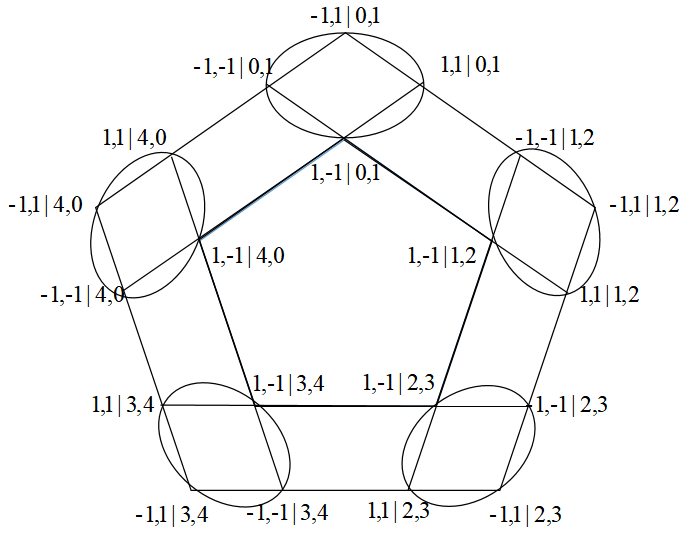}
    \caption{The exclusivity graph of KCBS experiment}\label{2014}
\end{figure}

In Fig.\ref{2014}, $x,y|i,j$ corresponds the event $(\hat{A}_i=x,\hat{A}_j=y)$, and the same straight line or circumference represent sets of pairwise exclusive events. However, for KCBS experiment, the vertices $-1,-1|i,i+1$ are impossible, and the vertices $-1,1|i,i+1$, $1,-1|i-1,i$ are equivalent. It means that Fig.\ref{2014} shows a deviation in depicting KCBS experiment.\par

Such a issue can also explain why the exclusivity graph approach does not consider local consistency. As in the case of sheaf theory approach, $\hat{A}=a_0|0\rangle\langle0|+a_1|1\rangle\langle1|+a_2|2\rangle\langle2|$ and $\hat{B}=b_0|0\rangle\langle0|+b_1|x\rangle\langle x|+b_2|y\rangle\langle y|$ are not commeasurable but $\hat{A}=a_0$ is equivalent to $\hat{B}=b_0$. The exclusivity graph $G$ consists of two disjoint cliques $C_A=\{(\hat{A}=a_0),(\hat{A}=a_1),(\hat{A}=a_2)\}$ and $C_B=\{(\hat{B}=b_0),(\hat{B}=b_1),(\hat{B}=b_2)\}$. Thus any probability distributions $p$ on $G$ are a model, but $p$ is not locally consistent when $p(\hat{A}=a_0)\neq p(\hat{B}=b_0)$. In reality, one cannot determine if a model on an exclusivity graph is locally consistent without additional information of the events.\par

The exclusivity graph approach can depict correctly the scenarios only when all the observables are complete and do not lead to equivalent events. Therefore, the applicability of exclusivity graph is limited to the scenarios generated by rank-1 projectors.

\section{Event-based contextuality theory}\label{Event-based contextuality theory}

We have shown that the observable-based contextuality theories encounter challenges when describing the scenarios composed of general observables. This is the reason why these known theories primarily focus on rank-1 scenarios. One can refine these theories by introducing additional constrains, but we believe a more essential approach is to substitute the event-based definition for the observable-based one.\par

Let's consider the classical probability theory. The event-algebra is derived from the sample space, and the observable (that is, the random variable) is defined by the events. Therefore, the fundamental structure of a probability theory is the event-algebra, not observable. It is especially true for quantum mechanics, where the hidden-variable spaces do not exist\citep{Kochen1967The}. \par

In the section, we establish the event-based contextuality theory.

\subsection{Exclusive partial Boolean algebra}

Hidden-variable theory is based on Boolean algebra. The well-known axiomatized classical probability theory is based on $\sigma$ algebra, which is also a type of Boolean algebra. Boolean algebra characterizes the logic of classical world, so it can be utilized to describe the single context in contextuality theory.\par

In quantum mechanics, multiple contexts correspond to multiple Boolean algebras. Kochen and Specker introduce the partial Boolean algebra to explain how these Boolean algebras form the global system. The definition below is from \cite{Van2012Noncommutativity} and \cite{Abramsky2020The}.

\begin{definition}
    If $B$ is a set with
    \begin{itemize}
    \item a reflexive and symmetric binary relation $\odot\subseteq B\times B$,
    \item a (total) unary operation $\lnot:\ B\rightarrow B$,
    \item two (partial) binary operations $\land,\ \lor:\ \odot\rightarrow B$,
    \item elements $0,1 \in B$,
    \end{itemize}
    satisfying that for every subset $S\subseteq B$ such that $\forall a,\ b\in S,\ a\odot b$, there exists a Boolean subalgebra $C\subseteq B$ determined by $(C,\land,\lor,\lnot,0,1)$ and $S\subseteq C$, then $B$ is called a \textbf{partial Boolean algebra}, written by $(B,\odot)$, or $(B,\odot;\land,\lor,\lnot,0,1)$ for details.\par
    Use $pBA$ to denote the collection of all partial Boolean algebras.
\end{definition}

The abbreviation $pBA$ is adopted from \cite{Abramsky2020The}, where $pBA$ represents the category of partial Boolean algebras. The relation $\odot$ represents the compatibility. $a\odot b$ if and only if $a,b$ belong to a Boolean subalgebra. Therefore, the contexts of scenario will be defined later by the Boolean subalgebras of $B$.\par

To depict the probability distributions on $pBA$, the notion state is defined below.

\begin{definition}
If $B\in pBA$,  then a \textbf{state} on $B$ is defined by a map $p:B\rightarrow[0,\ 1]$ such that
 \begin{itemize}
    \item $p(0)=0$.
    \item $p(\neg x)=1-p(x)$.
    \item for all $x,y\in B$ with $x\odot y$,$\ p(x\lor y)+p(x\land y)=p(x)+p(y)$.
 \end{itemize}
A state is called a \textbf{0-1 state} if its range is $\{0,1\}$. Use $s(B)$ to denote the states set on $B$, and $s_{01}(B)$ the 0-1 states set of $B$.\label{def_state}
\end{definition}

A 0-1 state is exactly a homomorphism from $B$ to $\{0,1\}$, that is, a truth-values assignment. Definition \ref{def_state} only requires finite additivity because $pBA$ only demands closure under finite unions.

Partial Boolean algebra was introduced to prove the well-known Kochen-Specker theorem\citep{Kochen1967The}. The most significant advantage of $pBA$ is that it captures the local consistency. If $B\in pBA$, $e\in B$ and $p\in s(B)$, we have $p_C(e)=p|_C(e)=p(e)$ for all the Boolean subalgebras $C$. Conversely, if $\{p_C\}$ are probability distributions satisfying local consistency, then there exists $B\in pBA$ and $p\in s(B)$ such that $p|_C=p_C$ for all contexts $C$. Therefore, the local consistency (including the no-signalling) is actually the feature of quantum system's structure, not just the quantum states. In quantum mechanics, each projector $\hat{P}$ onto a Hilbert space $\mathcal{H}$ is constant, which deduces that the probability $p(\hat{P})$ is independent with the contexts, no matter if $p$ is a quantum state or not.

However, $pBA$ may be really unnatural. To explain it, we introduce the definition of exclusive events in $pBA$.

\begin{definition}
If $B\in pBA$, $a,b\in B$, then,
 \begin{itemize}
    \item $a\leq b$ is defined by $a\odot b$ and $a\land b=a$;
    \item $a,b$ are said to be exclusive, written $a\bot b$, if there exists an element $c\in B$ such that $a\leq c$ and $b\leq\neg c$.
 \end{itemize}
 \label{def-exclusive}
\end{definition}

The definition \ref{def-exclusive} encompasses the definition in observable-based theories. For example, consider two (observable-based) events $a=(A_1=0,A_2=0)$ and $b=(A_2=1,A_3=0)$. We have $a\leq(A_2=0)$ and $b\leq(A_2=1)$. Thus $a\bot b$. Notice that the definition \ref{def-exclusive} is independent of the forms of observables. If $a=(A_1=0,A_2=0)$ and $b=(B_1=0,B_2=0)$, but $(A_1=0)$ is exclusive to $(B_1=0)$, then $a$ and $b$ are also exclusive.

Next we give an example of ``unnatural" $pBA$. Consider $B_1$ and $B_2\in pBA$. They are similar to each other and illustrated in the Fig.\ref{pB1} and Fig.\ref{pB2}, where the lines represent the partial orders incompletely.

\begin{figure}[H]
      \begin{minipage}[t]{0.5\linewidth}
          \centering
          \includegraphics[width=0.9\linewidth]{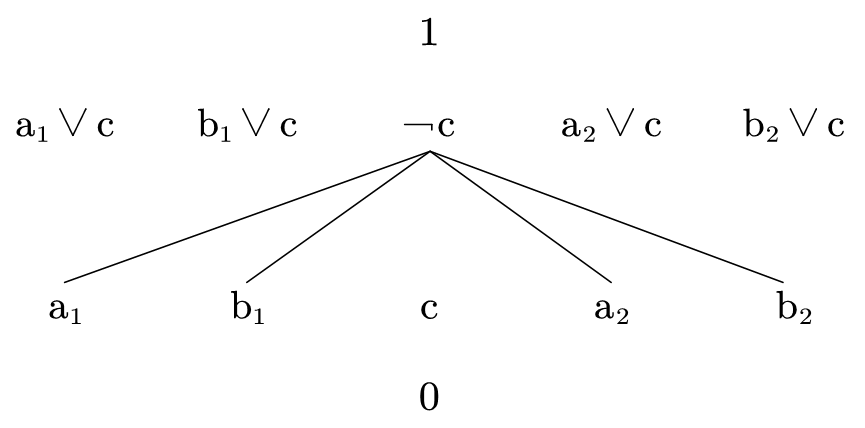}
          \caption{$B_1$}\label{pB1}
      \end{minipage}
      \begin{minipage}[t]{0.5\linewidth}
          \includegraphics[width=0.9\linewidth]{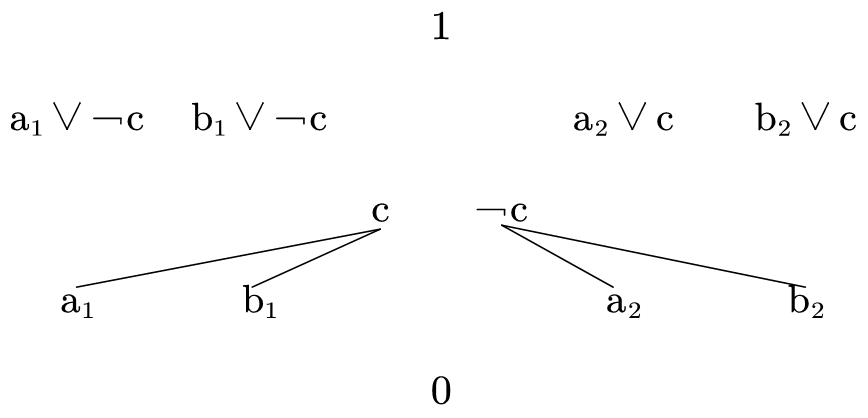}
          \caption{$B_2$}\label{pB2}
      \end{minipage}
\end{figure}

$B_1$ and $B_2$ have the same elements. $B_2$ contains two maximal Boolean subalgebras $C_1$ and $C_2$, which are generated by $\{a_1,b_1,\neg c\}$ and $\{a_2,b_2,c\}$. We have $a_1\leq c$ and $c\leq a_2\lor c$, but no $a_1\leq a_2\lor c$ because $a_1$ and $a_2\lor c$ are not in a same Boolean subalgebra. In other words, $B_2$ does not satisfy transitivity, while $B_1$ does not present such issue.\par

To exclude the $pBA$ like $B_2$. Abramsky and Barbosa\citep{Abramsky2020The} generalize the exclusivity principle to $pBA$.

\begin{definition}
 $B\in pBA$ is said to satisfy \textbf{logical exclusivity principle (LEP)} or to be \textbf{exclusive} if $\bot\subseteq\odot$.\par
 Use $epBA$ to denote the collection of all exclusive partial Boolean algebras.
\end{definition}

In other words, $B$ is exclusive if $a\odot b$ for all $a\bot b$ in $B$. $epBA$ is an extremely important structure because it satisfies the following three propositions.

\begin{proposition}[Abramsky and Barbosa\citep{Abramsky2020The}]
If $B\in pBA$, then $B$ is exclusive iff $B$ is transitive.\label{prop1}
\end{proposition}

The proposition \ref{prop1} indicates that LEP is equivalent to the transitivity. Therefore, LEP is necessary for the capability of normal logic reasoning of any contextual systems.

\begin{proposition}[Abramsky and Barbosa\citep{Abramsky2020The}]
If $B\in epBA$, then any $p\in s(B)$ satisfy the exclusivity principle.\label{prop2}
\end{proposition}

The proposition \ref{prop2} indicates that LEP is stronger than the exclusivity principle. Note that quantum mechanics is exclusive, since any two exclusive projectors must be compatible. Therefore, LEP is more suitable than the exclusivity principle to describe the quantum mechanics.\par

We need to define several notions for the third proposition.

\begin{definition}
 Let $B\in pBA$.
   \begin{itemize}
    \item  The \textbf{atom} of $B$ is the element $a\in B$ such that $a\neq 0$, and $x\leq a$ implies $x=0$ or $x=a$ for any $x\in B$. Use $A(B)$ to denote the atoms set of $B$.
    \item The \textbf{atom graph} of $B$, written $AG(B)$, is defined by a graph with vertices set $A(B)$ and $a_1,a_2\in A(B)$ are adjacent iff $a_1\odot a_2$ and $a_1\neq a_2$ .
 \end{itemize}
\end{definition}

\begin{definition}
If $G$ is a finite simple graph, a \textbf{state} on $G$ is defined by a map $p:V(G)\rightarrow[0,\ 1]$ such that for each maximal clique $C$ of $G$, $\sum\limits_{v\in C}p(v)=1$. Use $s(G)$ to denote the states set on $G$.
\end{definition}

\begin{proposition}[Liu et al.\citep{Liu2025Atom}]
 If $B,B_1$ and $B_2$ are any finite $epBA$. then
  \begin{itemize}
    \item $s(B)\cong s(AG(B))$.
    \item $B_1\cong B_2$ iff $AG(B_1)\cong AG(B_2)$.
 \end{itemize}\label{prop3}
\end{proposition}

The proposition \ref{prop3} reveals that each finite $epBA$ is determined by its atom graph. Therefore, the finite quantum systems can be exactly depicted by the graph theory. In reality, this proposition establishes a rigorous mathematical foundation for the graph-theoretic approaches within the quantum contextuality field. If $B\in epBA$ is finite, the atom graph $AG(B)$ is the formalization of the exclusivity graph, and $AG(B)$ preserves all the information from $B$ without loss due to the proposition \ref{prop3}.\par

Based on the reasons outlined above, we deem $epBA$ to be a suitable structure for the general probability theory encompassing quantum mechanics within the contextuality research. We present the event-based contextuality theory with $epBA$ in the next subsection.

\subsection{Contextuality theory based on $epBA$}

In classical probability theory with the sample space $\Omega$, the events are defined by elements of the power set algebra $\mathcal{P}(\Omega)$. In quantum systems with Hilbert space $\mathcal{H}$, the events are the elements of the projectors algebra (see in Sect. \ref{sec-quantum-contextuality}). In general cases, we would have a given $epBA$, which contains all concerned events.

\begin{definition}
$B\in epBA$ are called \textbf{scenarios}, and the elements $e\in B$ are called \textbf{events}.
\end{definition}

In observable-based theory, $X=\{A_1,...,A_n\}$ generates different scenarios corresponding to different relations among the outcomes of $A_1,...,A_n$, including the compatibility, equivalence, exclusivity and logical relations. The observable-based scenarios cannot characterize all of these relations, which are actually embodied by the structure of $epBA$.

The hidden-variable theory is described by the Boolean algebra, which reflects the local aspects of the contextual systems. That is what we call contexts.

\begin{definition}
The (maximal) Boolean subalgebras of $B\in epBA$ are called the \textbf{(maximal) contexts} of $B$.
\end{definition}

In observable-based theory, the contexts are defined by the compatibility between observables, which inevitably loses some relations between the contexts.\par

A context of $B\in epBA$ exactly represents an observable of the system, and a maximal context represent a complete observable. If $B$ is finite, its maximal context $C$ is a finite Boolean algebra generated by atoms $A(C)$, and $a\in A(C)$  represent the values of a complete observable.

If $B\in epBA$, the states $p\in s(B)$ satisfy the local consistency and exclusivity principle. We define ``model" of our theory to be the state, while we will continue to use ``state" in the following text to avoid confusion.

\begin{definition}
If $B\in epBA$, $p\in s(B)$ are called the \textbf{models} on $B$.
\end{definition}

Our definition of the contextuality can be considered as a promotion of the sheaf theory approach (definition \ref{sheaf3}) and the exclusivity graph approach (definition \ref{exclusivity4}).  The 0-1 states $\delta\in s_{01}(B)$ actually assign values to all the complete observables of $B$, thus $\delta$ formalize the global events in the sheaf theory approach. Then the global probability distributions $d$ are the convex combination of global events, that is, global deterministic probability distributions.

\begin{definition}
 Let $B\in epBA$.  The state $p$ is \textbf{non-contextual} if $p\in s_{NC}(B):=conv(s_{01}(B))$, where $conv$ represents the convex combination.\label{def_contextuality}
\end{definition}

Now we prove that the definition \ref{def_contextuality} is consistent with the NCHV models.\par

\begin{lemma}
Let $B\in pBA$, and $p\in s(B)$. $p(e_1)\leq p(e_2)$ if $e_1\leq e_2$.\label{p-monotonicity}
\end{lemma}
\begin{proof}
$e_1\leq e_2$ means that $e_1\land e_2=e_1$. Then $e_1\lor e_2=(e_1\land e_2)\lor e_2=(e_1\lor e_2)\land(e_2\lor e_2)=(e_1\lor e_2)\land e_2$, so $e_1\lor e_2\leq e_2$, thus $e_1\lor e_2=e_2$. Therefore, $p(e_2)=p(e_1\lor(e_2\land\neg e_1))=p(e_1)+p(e_2\land\neg e_1)-p(e_1\land e_2\land\neg e_1)=p(e_1)+p(e_2\land\neg e_1)\geq p(e_1)$.
\end{proof}

\begin{theorem}
Let $B\in epBA$. $p\in s_{NC}(B)$ iff there is a NCHV model which realizes $p$.\label{thm-contextuality}
\end{theorem}
\begin{proof}
Necessity. If $p\in s_{NC}(B)$, let $\Lambda=s_{01}(B)$. $\Lambda$ is not empty because $s_{NC}(B)$ is not empty. Then $p=\sum\limits_{\lambda\in\Lambda}p_{\Lambda}(\lambda)\lambda$. We have $p_{\Lambda}(\lambda)\geq 0$ and $\sum\limits_{\lambda\in\Lambda}p_{\Lambda}(\lambda)=1$. Thus $p_{\Lambda}$ is a probability distribution on $\Lambda$. For any $e\in B$ and $\lambda\in\Lambda$, we define a deterministic conditional probability $\delta(e|\lambda)=\lambda(e)$.\par

We verify that $\delta(e|\lambda)=\lambda(e)$ satisfies the factorizability property. Since $\lambda$ is a 0-1 state on $B$, for $e_1\odot e_2$, if $\lambda(e_1)=\lambda(e_2)=0$ then $\lambda(e_1\lor e_2)+\lambda(e_1\land e_2)=0$, so $\lambda(e_1\land e_2)=0$; if $\lambda(e_1)=\lambda(e_2)=1$ then $\lambda(e_1\lor e_2)+\lambda(e_1\land e_2)=2$, so $\lambda(e_1\land e_2)=1$; if $\lambda(e_1)=1$ and $\lambda(e_2)=0$ then $\lambda(e_1\lor e_2)+\lambda(e_1\land e_2)=1$, so $\lambda(e_1\land e_2)=0$ since $p(e_1\land e_2)\leq p(e_1\lor e_2)$ (the lemma \ref{p-monotonicity}); similarly to $\lambda(e_1)=0$ and $\lambda(e_2)=1$. In summary, $\lambda(e_1\land e_2)=\lambda(e_1)\lambda(e_2)$.\par

Therefore, for any contexts $C\subset B$ and $e\in C$, $p_C(e)=p|_C(e)=p(e)=\sum\limits_{\lambda\in\Lambda}p_{\Lambda}(\lambda)\delta(e|\lambda)$. $p$ can be realized by a NCHV model.\par

Sufficiency. If there exists $\Lambda$, $p_{\Lambda}$ and $\delta(\cdot|\lambda)$ such that $p(e)=\sum\limits_{\lambda\in\Lambda}p_{\Lambda}(\lambda)\delta(e|\lambda)$ is a NCHV model. $\delta(e|\lambda)$ satisfies the factorizability property. We prove that $\delta(\cdot|\lambda)\in s_{01}(B)$ for $p_{\Lambda}(\lambda)>0$.\par
$p(0)=\sum\limits_{\lambda\in\Lambda,p_{\Lambda}(\lambda)>0}p_{\Lambda}(\lambda)\delta(0|\lambda)=0$, thus $\delta(0|\lambda)=0$.\par
$\delta(e\land\neg e|\lambda)=\delta(0|\lambda)=0=\delta(e|\lambda)\delta(\neg e|\lambda)$. Suppose there exist $e\in B$ and $\lambda'\in\Lambda$ $(p_{\Lambda}(\lambda')>0)$ such that $\delta(e|\lambda')=\delta(\neg e|\lambda')=0$, then $p(e)+p(\neg e)=\sum\limits_{\lambda\in\Lambda,\lambda\neq\lambda'}p_{\Lambda}(\lambda)(\delta(e|\lambda)+\delta(\neg e|\lambda))<1$, which induces a contradiction. Therefore, $\delta(e|\lambda)=0$, $\delta(\neg e|\lambda)=1$ or $\delta(e|\lambda)=1$, $\delta(\neg e|\lambda)=0$ for all $e\in B$ and $\lambda\in\Lambda$ $(p_{\Lambda}(\lambda)>0)$.\par
Finally, $\delta(e_1\lor e_2|\lambda)+\delta(e_1\land e_2|\lambda)=\delta(\neg(\neg e_1\land\neg e_2)|\lambda)+\delta(e_1|\lambda)\delta(e_2|\lambda)=1-\delta(\neg e_1|\lambda)\delta(\neg e_2|\lambda)+\delta(e_1|\lambda)\delta(e_2|\lambda)=1-((1-\delta(e_1|\lambda))(1-\delta(e_2|\lambda)))+\delta(e_1|\lambda)\delta(e_2|\lambda)
=\delta(e_1|\lambda)+\delta(e_2|\lambda)$.\par

In conclusion, $\delta(\cdot|\lambda)\in s_{01}(B)$ for $p_{\Lambda}(\lambda)>0$. Therefore, $p$ is a convex combination of 0-1 states. $p\in s_{NC}(B)$.
\end{proof}

The theorem \ref{thm-contextuality} formalizes the relevant conclusions in exclusivity graph approach and the sheaf theory approach. \par

We have established an event-based contextuality theory. It can handle scenarios generated by general events, and resolve the contradictions arising from the observable-based definitions. In the next section, we formalize several important quantum contextuality.

\section{Quantum contextuality}\label{sec-quantum-contextuality}

Contextuality theories are proposed to abstract the features of quantum mechanics, that is, quantum contextuality. In this section, we use our theory to formalize some important but confusing notions. Specifically speaking, we will formalize ``Bell nonlocality", ``KCBS contextuality", ``logical contextuality", ``strong contextuality", ``maximal contextuality", ``full contextuality", ``KS contextuality" and ``state-independent contextuality" and investigate their relationship.\par

All these types of contextuality can be defined on $epBA$, while some of them are valuable only within the quantum mechanics. For example, the state-independent contextuality, abbreviated SIC, specifically indicates that all quantum states present contextuality. For simplicity, we will focus on the quantum systems in this section.

\subsection{Quantum system}

Quantum mechanics points out that the world is essentially random, thus all the physical phenomena can be depicted by "event" and "the probability of event". An event refers to the outcome of a single measurement, or equivalently, the value of an observable. If the considered Hilbert space is $\mathcal{H}$, then an observable is depicted by a bounded self-adjoint operator $\hat{A}$ onto $\mathcal{H}$. An event is a proposition like $(\hat{A}\in\Delta)$ ($\Delta$ is a Borel set of $\mathbb{R}$), which is depicted by a projector $\hat{P}$ onto the corresponding eigenspace of $\hat{A}$. Therefore, the event-algebra of quantum mechanics is the projector algebra.\par

In 1936, Birkhoff and Von Neumann introduced the standard quantum logic based on this perspective\citep{Birkhoff1936The}. Let $P(\mathcal{H})$ denote the set of projectors onto $\mathcal{H}$. If $\hat{P}_1,\ \hat{P}_2$ are projectors onto $S_1,\ S_2$, $\hat{P}_1\land\hat{P}_2$ is defined to be the projector onto $S_1\cap S_2$, and $\neg\hat{P}_1$ is defined to be the projector onto $S_1^{\bot}$. Then $\hat{P}_1\lor\hat{P}_2=\neg(\neg\hat{P}_1\land\lnot\hat{P}_2)$. If these operations are totally defined, then $P(\mathcal{H})$ is an orthocomplemented modular lattice, which presents several disadvantages such as not satisfying the distributive law \citep{Doering2010Topos}. Therefore, Kochen and Specker adopt $P(\mathcal{H})$ as a partial Boolean algebra in 1960s\citep{Kochen1967The}. Define $\hat{P}_1\odot\hat{P}_2$ if $\hat{P}_1\hat{P}_2=\hat{P}_2\hat{P}_1$, that is, the compatibility relation. We have $P(\mathcal{H})=(P(\mathcal{H}),\odot;\land,\lor,\neg,\hat{0},\hat{1})$ is a partial Boolean algebra, where $\hat{0}$ is the zero projector, and $\hat{1}$ is the projector onto $\mathcal{H}$. A context ,or a Boolean subalgebra of $P(\mathcal{H})$, consists of mutually compatible projectors. Any quantum system on $\mathcal{H}$ can be treated as partial Boolean subalgebra of $P(\mathcal{H})$.\par

\begin{definition}[\cite{Liu2025Atom}]
A \textbf{quantum system} is defined by a partial Boolean subalgebra of $P(\mathcal{H})$ for some Hilbert space $\mathcal{H}$. Use $QS$ to denote the collection of all quantum systems.
\end{definition}

If $\hat{P}_1\bot\hat{P}_2$, in other words, there is a projector $\hat{P}$ such that $\hat{P}_1\leq\hat{P}$ and $\hat{P}_2\leq\neg\hat{P}$, then $\hat{P}_1,\hat{P}_2$ must be orthogonal, so $\hat{P}_1\odot\hat{P}_2$. Therefore, quantum systems are exclusive.  We have

$$QS\subseteq epBA\subseteq pBA.$$

The quantum states $\rho$ on $Q\in QS$ are the density operators on $\mathcal{H}$. We simply treat $\rho$ as a state on $Q$, that is, $\rho(\hat{P})=tr(\rho\hat{P})$. Use $qs(Q)$ to denote the set of quantum states on $Q$. It is easy to verify that $qs(Q)\subseteq s(Q)$.

\subsection{Bell nonlocality}

Bell nonlocality\citep{Bell1964On} is the earliest contextuality, which has long held a dominant position in the field of quantum applications. In fact, Bell nonlocality is a term that specifically refers to the contextuality of the Bell systems.\par

A Bell system consists of $n$ spacelike-separated parts (Hilbert spaces like qubits), $k$ incompatible observables on each part and $l$ eigenvalues for each observable, which is called a $(n,k,l)$ system. The simplest Bell system is the Clauser-Horne-Shimony-Holt (CHSH) $(2,2,2)$ system\citep{Clauser1969Proposed}, which consists of two qubits with the Pauli operators $Z_1$, $X_1$ on the first qubit, and $S_2=-\frac{1}{\sqrt{2}}(Z_2+X_2)$, $T_2=\frac{1}{\sqrt{2}}(Z_2-X_2)$ on the second one. There are totally 16 atom event, $(Z_1=\pm 1,S_2=\pm 1)$, $(Z_1=\pm 1,T_2=\pm 1)$, $(X_1=\pm 1,S_2=\pm 1)$ and $(X_1=\pm 1,T_2=\pm 1)$. Each of them corresponds a 4-dimensional rank-1 projector. These projectors generate the CHSH system $Q_{CHSH}\in QS$. Note that the Bell systems are finite, so we can use the atom graph $AG(Q_{CHSH})$ to describe the system due to the proposition \ref{prop3}. $AG(Q_{CHSH})$ is shown in Fig. \ref{CHSH} which is equal to the exclusivity graph in \cite{Adan2014Graph}. The straight line or a circumference in Fig. \ref{CHSH} represent a maximal clique, that is, a maximal context. The $\hat{P}_{++|ZS}$, and similar, represents the projector corresponding the event $(Z_1=+1,S_2=+1)$.

\begin{figure}[H]
    \centering
    \includegraphics[width=0.5\linewidth]{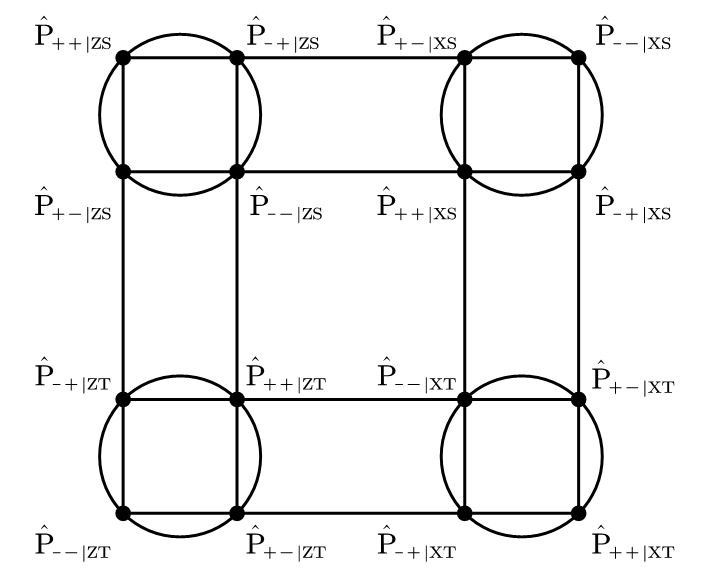}
    \caption{The atom graph of CHSH system $AG(Q_{CHSH})$.}
    \label{CHSH}
\end{figure}

Generally, the $(n,k,l)$ system $Q_{(n,k,l)}\in QS$ has $(kl)^n$ atoms, so $AG(Q_{n,k,l})$ has $(kl)^n$ vertices.\par

\begin{definition}
For Bell system $Q_{(n,k,l)}\in QS$ and $p\in s(Q_{(n,k,l)})$, $p$ is said to be \textbf{Bell nonlocal} if $p$ is contextual, that is, $p\notin s_{NC}(Q_{(n,k,l)})$.
\end{definition}

A state $p$ is verified to be Bell nonlocal usually by the Bell inequalities. The proposition below shows one version of the CHSH inequalities\citep{Clauser1969Proposed}. \par

\begin{proposition}
If $p\in s_{NC}(Q_{CHSH})$, then
\begin{equation}
\begin{split}
S_{CHSH}(p)=&p(\hat{P}_{++|ZS})+p(\hat{P}_{--|ZS})+p(\hat{P}_{++|XT})+p(\hat{P}_{--|XT})\\
+&p(\hat{P}_{++|XS})+p(\hat{P}_{--|XS})+p(\hat{P}_{+-|ZT})+p(\hat{P}_{-+|ZT})\leq3
\nonumber
\end{split}
\end{equation}
\end{proposition}

Consider the entangled state $|\psi\rangle=\frac{1}{\sqrt{2}}(|01\rangle-|10\rangle)$ and $\rho=|\psi\rangle\langle\psi|\in qs(Q_{CHSH})$. Calculations yield that $S_{CHSH}(\rho)=2+\sqrt{2}>3$. Thus $\rho\notin s_{NC}(Q_{CHSH})$. $\rho$ is called a Bell state.

\subsection{KCBS contextuality}

The significant KCBS experiment\citep{Alexander2008Simple} has been introduced in the subsection \ref{subsec-ex-graph}. The observables $\{\hat{A}_0,...,\hat{A}_4\}$ induces exactly ten 3-dimensional rank-1 projectors $\hat{P}_i$ and $\hat{P}_{i\ i+1}$ corresponding to $(\hat{A}_i=-1)$ and $(\hat{A}_i=1, \hat{A}_{i+1}=1)$ $(i=0,1,2,3,4)$. Note that $(\hat{A}_i=-1)$, $(\hat{A}_{i-1}=1,\hat{A}_i=-1)$ and $(\hat{A}_i=-1,\hat{A}_{i+1}=1)$ are equivalent, and $(\hat{A}_i=-1,\hat{A}_{i+1}=-1)$ are impossible. In another way, the KCBS system $Q_{KCBS}$ is generated by projectors $\{\hat{P}_i\}_{i=0}^4$ and $\hat{P}_{i\ i+1}=\neg(\hat{P}_i\lor\hat{P}_{i+1})$ (with the sum modulo 5). The atom graph of $Q_{KCBS}$ is shown in Fig. \ref{KCBS}.

\begin{figure}[H]
    \centering
    \includegraphics[width=0.5\linewidth]{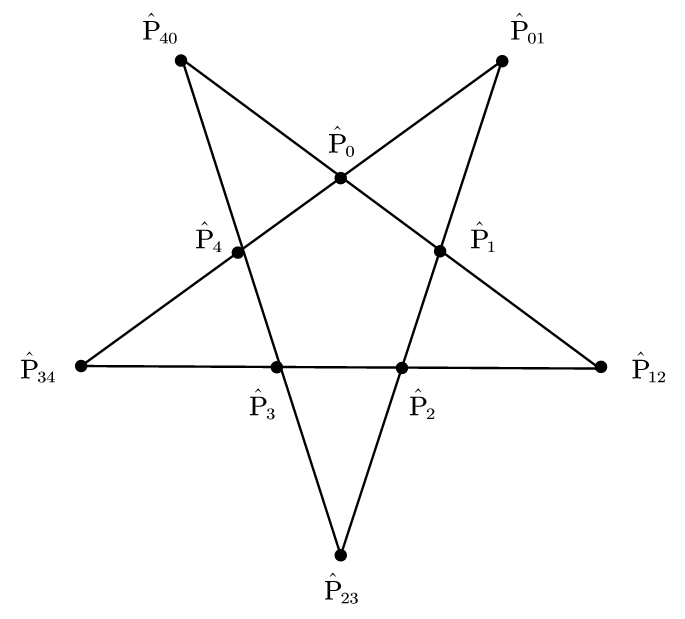}
    \caption{The atom graph of KCBS system $AG(Q_{KCBS})$}
    \label{KCBS}
\end{figure}

Compared with the exclusivity graph Fig. \ref{2014}, which has no realization in 3-dimensional space. The Fig. \ref{KCBS} correctly depicts the scenario of KCBS experiment.\par

To prove that $Q_{KCBS}$ allows contextuality, one can use the following KCBS inequality\citep{Alexander2008Simple,Adan2010Contextuality}.

\begin{proposition}
If $p\in s_{NC}(Q_{KCBS})$, then
$$S_{KCBS}(p)=p(\hat{P}_0)+p(\hat{P}_1)+p(\hat{P}_2)+p(\hat{P}_3)+p(\hat{P}_4)\leq2$$
\end{proposition}

If $\hat{P}_i$ correspond to the unit-vectors $v_i$ $(i=0,1,2,3,4)$, let $v=v_0+v_1+v_2+v_3+v_4$. Then the pure state $\rho\in qs(Q_{KCBS})$ defined by $v/||v||$ violates the KCBS inequality. In fact, $S_{KCBS}(\rho)=\sqrt{5}>2$. $\rho$ is the KCBS state.

\subsection{Logical contextuality}

Both the Bell state and KCBS state witness the standard contextuality depicted by definition \ref{def_contextuality}. Next we introduce some stronger type of contextuality.

Abramsky and Brandenburger established a hierarchy of contextuality with sheaf theory approach\citep{Abramsky2011sheaf}, where logical contextuality (or possibilistic contextuality) was a part of it. Silva redefined this notion with exclusivity graph\citep{Silva2017Graph}. The definitions below is consistent with theirs.

\begin{definition}
Let $Q\in QS$. $p\in s(Q)$ is said to be \textbf{logically contextual} if there exists an event $e\in Q$ such that $p(e)>0$, and for any $\delta\in s_{01}(Q)$ satisfying $\delta(e)=1$, there is another event $e'\in Q$ such that $\delta(e')=1$ but $p(e')=0$.
\end{definition}

The logical contextuality is introduced to abstract the Hardy-type proof of KS theorem\citep{Hardy1993Nonlocality}, which is also an elegant method to witness the Bell nonlocality without using inequalities. We have the logical contextuality is strictly stronger than contextuality.

\begin{theorem}[\citep{Abramsky2011sheaf}]
If $p\in s(Q)$ is logically contextual, then it is contextual.\label{logical-contextual}
\end{theorem}
\begin{proof}
Suppose that $p\in s_{NC}(Q)$, then $p=\sum\limits_{\delta\in s_{01}(Q)}w_{\delta}\delta$ where $w_{\delta}\geq0$ and $\sum\limits_{\delta\in s_{01}(Q)}w_{\delta}=1$. For any $e\in Q$ such that $p(e)>0$, there is at least one $\delta'\in s_{01}(Q)$ such that $w_{\delta'}>0$ and $\delta'(e)=1$. Then for any other event $e'$, $\delta'(e')=1$ implies $p(e')>0$ , which is contradictory to $p$ is logically contextual. Thus $p\notin s_{NC}(Q)$.
\end{proof}

The converse of theorem \ref{logical-contextual} is not true. Consider the KCBS system $Q_{KCBS}$ and the KCBS state $\rho$. $\rho$ is contextual, but it is not logically contextual. In fact, $\rho(\hat{P}_i)=\sqrt{5}/5$ and $\rho(\hat{P}_{i\ i+1})=1-2\sqrt{5}/5>0$ for $i=0,1,2,3,4$. Thus $\rho(e)>0$ for any $e\in Q_{KCBS}$. Obviously $\rho$ is not logically contextual.

A known example of logically contextuality is the Hardy state\citep{Hardy1993Nonlocality}. It is a quantum state $\rho_{Hardy}$ on the CHSH system, shown in the Fig. \ref{Hardy} where the vertice 0 represents $\rho_{Hardy}(e)=0$, and the vertices 1 represents $\rho_{Hardy}(e)>0$.\par

\begin{figure}[H]
\centering
\includegraphics[width=0.5\linewidth]{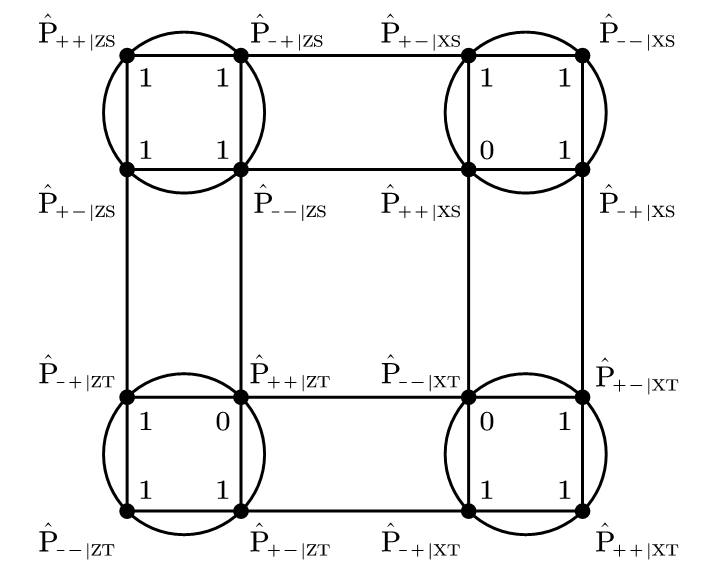}
\caption{The Hardy state $\rho_{Hardy}\in qs(Q_{CHSH})$}
\label{Hardy}
\end{figure}

For any $\delta\in s_{01}(Q_{CHSH})$ such that $\delta(\hat{P}_{++|ZS})=1$, to avoid contradiction, $\delta(\hat{P}_{-+|XS})$ and $\delta(\hat{P}_{+-|ZT})$ must be 1, which forces $\delta(\hat{P}_{--|XT})$ to be 1. However, $\rho_{Hardy}(\hat{P}_{--|XT})=0$. Thus the Hardy state $\rho_{Hardy}$ is logically contextual. In other words, if $Z_1,S_2=+1,+1$, then $X_1=-1$ and $T_2=-1$, which is contradictory to $\rho_{Hardy}$. \par

\subsection{Strong contextuality}

Strong contextuality was also presented by Abramsky and Brandenburger.\citep{Abramsky2011sheaf}. Both it and logical contextuality are introduced to describe the contextuality in terms of logic rather than probability. Silva also redefined the strong contextuality with exclusivity graph\citep{Silva2017Graph}. The definitions below is consistent with the definition \ref{sheaf4}, where the 0-1 states $\delta$ play the role of global events.

\begin{definition}
Let $Q\in QS$. $p\in s(Q)$ is said to be \textbf{strongly contextual} if for any $\delta\in s_{01}(Q)$, there is an event $e\in Q$ such that $\delta(e)=1$ but $p(e)=0$.
\end{definition}

The strong contextuality is introduced to abstract the GHZ-type proof of KS theorem\citep{Greenberger1989Going}, which is a stronger version of Hardy-type proof. We have the strong contextuality is strictly stronger than logical contextuality.

\begin{theorem}[\cite{Abramsky2011sheaf}]
If $p\in s(Q)$ is strongly contextual, then it is logically contextual.\label{strong-logical}
\end{theorem}
\begin{proof}
Derive from the relevant definitions.
\end{proof}

The converse of theorem \ref{strong-logical} is not true. In fact, the Hardy state $p_{Hardy}$ is not strongly contextual, because one can easily find a $\delta\in s_{01}(Q_{CHSH})$ (for example, let $\hat{P}_{--|ZS}$, $\hat{P}_{-+|ZT}$, $\hat{P}_{+-|XS}$ and $\hat{P}_{++|XT}$ be assigned 1) which is not contradictory with $\rho_{Hardy}$ for any $e\in Q_{CHSH}$.

The GHZ state\citep{Greenberger1989Going} $|\psi_{GHZ}\rangle=\frac{1}{\sqrt{2}}(|000\rangle+|111\rangle)$ is a famous example of strong contextuality. On a $(3,2,2)$ Bell system, $\rho_{GHZ}=|\psi_{GHZ}\rangle\langle\psi_{GHZ}|$ is shown in the Table \ref{GHZ}, where $X_i,Y_i$ $(i=1,2,3)$ are the Pauli operators.\par

\begin{table}[h]
\renewcommand{\arraystretch}{1.5}
\begin{tabular}{c|c|c|c|c|c|c|c|c|}
             & $+++$ & $++-$ & $+-+$ & $+--$ & $-++$ & $-+-$ & $--+$ & $---$ \\ \hline
 $X_1X_2X_3$ & $\frac{1}{4}$ & $0$ & $0$ & $\frac{1}{4}$ & $0$ & $\frac{1}{4}$ & $\frac{1}{4}$ & $0$ \\ \hline
 $X_1X_2Y_3$ & $\frac{1}{8}$ & $\frac{1}{8}$ & $\frac{1}{8}$ & $\frac{1}{8}$ & $\frac{1}{8}$ & $\frac{1}{8}$ & $\frac{1}{8}$ & $\frac{1}{8}$ \\ \hline
 $X_1Y_2X_3$ & $\frac{1}{8}$ & $\frac{1}{8}$ & $\frac{1}{8}$ & $\frac{1}{8}$ & $\frac{1}{8}$ & $\frac{1}{8}$ & $\frac{1}{8}$ & $\frac{1}{8}$ \\ \hline
 $X_1Y_2Y_3$ & $0$ & $\frac{1}{4}$ & $\frac{1}{4}$ & $0$ & $\frac{1}{4}$ & $0$ & $0$ & $\frac{1}{4}$ \\ \hline
 $Y_1X_2X_3$ & $\frac{1}{8}$ & $\frac{1}{8}$ & $\frac{1}{8}$ & $\frac{1}{8}$ & $\frac{1}{8}$ & $\frac{1}{8}$ & $\frac{1}{8}$ & $\frac{1}{8}$ \\ \hline
 $Y_1X_2Y_3$ & $0$ & $\frac{1}{4}$ & $\frac{1}{4}$ & $0$ & $\frac{1}{4}$ & $0$ & $0$ & $\frac{1}{4}$ \\ \hline
 $Y_1Y_2X_3$ & $0$ & $\frac{1}{4}$ & $\frac{1}{4}$ & $0$ & $\frac{1}{4}$ & $0$ & $0$ & $\frac{1}{4}$ \\ \hline
 $Y_1Y_2Y_3$ & $\frac{1}{8}$ & $\frac{1}{8}$ & $\frac{1}{8}$ & $\frac{1}{8}$ & $\frac{1}{8}$ & $\frac{1}{8}$ & $\frac{1}{8}$ & $\frac{1}{8}$ \\ \hline
 \end{tabular}
  \caption{The probability distributions induced by GHZ state $\rho_{GHZ}$}\label{GHZ}
\end{table}

$\rho_{GHZ}$ is strongly contextual. For any $\delta\in s_{01}(Q_{(3,2,2)})$, if $\delta(\hat{P}_{+++|X_1X_2X_3})=1$, then in order to avoid contradictory, $\delta$ must assign different values to $Y_1$, $Y_2$ and $Y_3$. In other words, $\delta(\hat{P}_{+|Y_1})=1$ iff $\delta(\hat{P}_{-|Y_2})=1$ iff $\delta(\hat{P}_{+|Y_3})=1$ iff $\delta(\hat{P}_{-|Y_1})=1$, which induces a contradictory. Similarly, one can verify that $\delta(\hat{P}_{+--|X_1X_2X_3})=1$, $\delta(\hat{P}_{-+-|X_1X_2X_3})=1$ or $\delta(\hat{P}_{--+|X_1X_2X_3})=1$ are all contradictory to $\rho_{GHZ}$.

\subsection{Maximal contextuality}

To quantify the strength of quantum contextuality, the contextual fraction is presented\citep{Abramsky2011sheaf,Amselem2012Experimental,Abramsky2017Contextual} which generalizes the notion of local fraction\citep{Elitzur1992Quantum,Aolita2012Fully} in the field of Bell nonlocality. A state is said to be maximally contextual if its contextual fraction equals to 1.

\begin{definition}
Let $Q\in QS$ and $p\in s(Q)$. Define $W_{NC}(p):=\max\{w\in[0,1]|p=wp_{NC}+(1-w)p_C\ s.t.\ p_{NC}\in s_{0,1}(Q)\ and\ p_C\in s(Q)\}$, called the \textbf{non-contextual fraction} of $p$. $W_C:=1-W_{NC}$ is called the \textbf{contextual fraction} of $p$. $p$ is said to be \textbf{maximally contextual} if $W_{C}(p)=1$.
\end{definition}

A major result gotten by the sheaf theory approach was that the maximal contextuality is equivalent to strong contextuality\citep{Abramsky2011sheaf}, but the original conclusion is slightly imprecise. We formalize their proof as follows.

\begin{theorem}\citep{Abramsky2011sheaf}
If $Q\in QS$ is finite, then $p\in s(Q)$ is maximally contextual if and only if it is strongly contextual.\label{strong-maximal}
\end{theorem}
\begin{proof}
Without loss of generality, let $s_{01}(Q)$ be not empty.

If $p$ is strongly contextual, suppose that $p=wp_{NC}+(1-w)p_C$, $p_{NC}=\sum\limits_{\delta\in s_{01}(Q)}w_{\delta}\delta$, thus $p=w\sum\limits_{\delta\in s_{01}(Q)}w_{\delta}\delta+(1-w)p_C$. The strong contextuality means that for any $\delta\in s_{01}(Q)$, there exists an event $e\in Q$ such that $\delta(e)=1$ but $p(e)=0$, which deduces that all the coefficients $ww_{\delta}$ must be 0. Therefore $w=0$, so $W_{NC}(p)=0$. $p$ is maximally contextual.\par

Conversely, when $p$ is maximally contextual, if there exists a 0-1 state $\delta$ such that for any $\delta(e)=1$, $p(e)>0$, let $\epsilon=\min\{p(e)|\delta(e)=1,e\in Q\}$, Obviously $0<\epsilon<1$. Then $p$ can be decomposed into $p=\epsilon\delta+(1-\epsilon)p'$, where $p'=\displaystyle\frac{1}{1-\epsilon}(p-\epsilon\delta)$. $p(e)\geq 0$ and for any maximal cliques $C$ of the atom graph $AG(Q)$, $\sum\limits_{e\in C}p'(e)=\displaystyle\frac{1}{1-\epsilon}(\sum\limits_{e\in C}p(e)-\epsilon\sum\limits_{e\in C}\delta(e))=1$. Therefore, $p'|_{A(Q)}\in s(AG(Q))$, so $p'\in s(Q)$ due to the proposition \ref{prop3}. It is contradictory to $p$ is maximally contextual. In conclusion, $p$ is strongly contextual.
\end{proof}

The proof is only valid for finite quantum systems $Q$. If $Q$ is infinite, the $\epsilon=\inf\{p(e)|\delta(e)=1,e\in Q\}$ may be 0 in the proof of necessity. Therefore, for general $Q\in QS$, we only have that the strong contextuality is stronger than the maximal contextuality for now.

\begin{theorem}
If $p\in s(Q)$ is strongly contextual, then it is maximally contextual.
\end{theorem}
\begin{proof}
Similar to the proof of sufficiency in theorem \ref{strong-maximal}
\end{proof}

For the same reasons, it is also unknown whether the maximal contextuality and logical contextuality can be compared in general cases. If $p(e)>0$ for all $e\in Q$, $p$ is not logically contextual,  but it is not clear whether $p$ can be maximally contextual when $\inf\{p(e)|e\in Q\}=0$.

\subsection{Full contextuality}

The most renowned method to witness quantum contextuality experimentally is the noncontextuality (NC) inequalities, such as the Bell inequalities and KCBS inequality introduced before. The NC inequalities can also measure the strength of contextuality.\par

\begin{definition}
Let $Q\in QS$ and $S(p)=\sum\limits_{i=1}^nw_ip(e_i)$, where $p\in s(Q)$, $w_i\in\mathbb{R}$ and $e_i\in Q$ $(i=1,...,n)$. If $S(p)\leq M$ for all $p\in s_{NC}(Q)$, then $S(p)\leq M$ is called a \textbf{NC inequality} on $Q$. Define $M_{NC}(S):=\max\limits_{p\in s_{NC}(Q)}S(p)$, $M_{Q}(S):=\max\limits_{p\in qs(Q)}S(p)$ and $M_{C}(S):=\max\limits_{p\in s(Q)}S(p)$.
\end{definition}

The terminologies ``full contextuality" and ``maximal contextuality" are often used interchangeably\citep{Amselem2012Experimental,Aolita2012Fully}. Here we define full contextuality if a state attains the upper bound of a NC inequality.

\begin{definition}
Let $Q\in QS$. $p\in s(Q)$ is said to be \textbf{fully contextual} if $S(p)=M_{C}(S)>M_{NC}(S)$ for some formula $S(p)=\sum\limits_{i=1}^nw_ip(e_i)$.
\end{definition}

The condition $M_{NC}(S)<M_C(S)$ is indispensable, otherwise $p$ can be not even contextual.

In fact, the full contextuality aims to depict the boundary of $s(Q)$. \cite{Adan2014Graph} shows that if we consider $p\in s(Q)$ as vectors, then $s_{NC}(Q)$ and $s(Q)$ are polytopes, and $qs(Q)$ is a convex set. We have $s_{NC}(Q)\subseteq s(Q)$ and $qs(Q)\subseteq s(Q)$, but we do not have $s_{NC}(Q)\subseteq qs(Q)$ because it is possible that some 0-1 states $\delta$ cannot be realized by the quantum states on $Q$. Nevertheless, Cabello et al. proved that $s_{NC}(Q)=qs(Q)$ iff $qs(Q)=s(Q)$ iff $Q$ does not present contextuality\citep{Adan2014Graph}.

The relationship between full contextuality and maximal contextuality has been partly revealed below\citep{Barrett2006Maximally}.

\begin{theorem}[\citep{Barrett2006Maximally}]
If $p\in s(Q)$ is fully contextual, then it is maximally contextual.\label{full-maximal}
\end{theorem}
\begin{proof}
If $s_{NC}(Q)$ is empty, it is obvious that $W_{NC}(p)=0$.\par

Otherwise, $S(p)=M_{C}(S)>M_{NC}(S)$ for some formula $S$. If $p=wp_{NC}+(1-w)p_C$, substituting the two sides of the equation into $S$, we have $S(p)=wS(p_{NC})+(1-w)S(p_C)=M_C(S)\leq wM_{NC}(S)+(1-w)M_C(S)$. Thus $w(M_C(S)-M_{NC}(S))\leq 0$. Therefore $w=0$, so $W_{NC}(p)=0$.
\end{proof}

Following the proof of theorem \ref{full-maximal}, one can get $W_{NC}(p)\leq\frac{M_C(S)-M_Q(S)}{M_C(S)-M_{NC}(S)}$. Therefore $M_Q(S)=M_C(S)$ can deduce maximal contextuality. The relevant work can be found in \citep{Barrett2006Maximally}.\par

It is not clear whether the converse of theorem \ref{full-maximal} is true or not. In other words, for any maximal contextual state $p$, whether there is a NC inequality being attained its upper bound $M_{C}(S)$ ($M_{C}(S)>M_{NC}(S)$)? In \cite{Aolita2012Fully}, Aolita, et al. claimed that for Bell systems, $W_{NC}(p)=0$ if and only if a Bell inequality was attained the upper bound, but they did not present the proof. The relationship between the full, maximal and strong contextuality need further research.

\subsection{KS contextuality}\label{KS contextuality subsection}
In 1967, Kochen and Specker proved that the quantum mechanics does not admit hidden variables by introducing a set of 3-dimensional rank-1 projectors which does not admit global true-value assignment\citep{Kochen1967The}, and such a system is said to satisfy KS contextuality.\par

KS contextuality is one of the most well-known contexuality, which characterizes the features of quantum systems, but not the state. Some theories such as Topos quantum logic even regards KS contextuality as the general contextuality\citep{Isham1998Topos}.

\begin{definition}
$Q\in QS$ is said to be \textbf{KS contextual} if $s_{01}(Q)$ is empty.\label{def-KS}
\end{definition}

Using the language of Kochen and Specker, the definition \ref{def-KS} means that there is no homomorphism from $Q$ to the Boolean algebra $\mathbbm{2}=\{0,1\}$. It is obviously that

\begin{theorem}
If $Q\in QS$ is KS contextual, all $p\in s(Q)$ are strongly contextual.\label{KS-strong}
\end{theorem}
\begin{proof}
Derive from the relevant definitions.
\end{proof}

In the field of KS-proof, the notion ``KS assignment" is more famous than KS contextuality itself. If $S$ is a set of rank-1 projectors (or vectors), the KS assignment on $S$ is defined below\citep{Budroni2022Kochen}.

\begin{definition}
Let $S$ be a set of $k$-dimensional rank-1 projectors. A \textbf{KS assignment} on $S$ is a function $f:S\to\{0,1\}$ satisfying
\begin{itemize}
    \item $f(\hat{P}_1)f(\hat{P}_2)=0$ if $\hat{P}_1,\hat{P}_2\in S$ are orthogonal (orthogonality).
    \item $\sum\limits_{i=1}^k f(\hat{P}_i)=1$ for mutually orthogonal $\{\hat{P}_i\}_{i=1}^k$ (completeness).
 \end{itemize}
\end{definition}

If $S$ does not admit KS assignments, it is called a ``KS set", which only contains rank-1 projectors or vectors. The earliest KS set introduced by Kochen and Specker consisted of 117 vectors\citep{Kochen1967The}. The proved minimal KS set is the CEG set\citep{Xu2020Proof} with 18 vectors. The 3-dimensional KS sets are also called ``KS systems". The minimal KS system has not been found, but it contains at least 23 vectors\citep{Li2022An}.

If $S$ is a KS set, the quantum system $Q$ generated by $S$ is KS contextual, because any $p\in s_{01}(Q)$ can be restricted to a KS assignment on $S$. However, the converse is not true. Recall the CEG set in Fig. \ref{CEG}. If we remove the vector $(1,0,0,0)$ from the CEG set to get a 17-vectors set, the new set and CEG set generate an identical KS contextual quantum system. However, the 17-vectors set is not a KS set, which has a KS assignment shown in Fig. \ref{CEG17}.

\begin{figure}[H]
    \centering
    \includegraphics[width=0.5\linewidth]{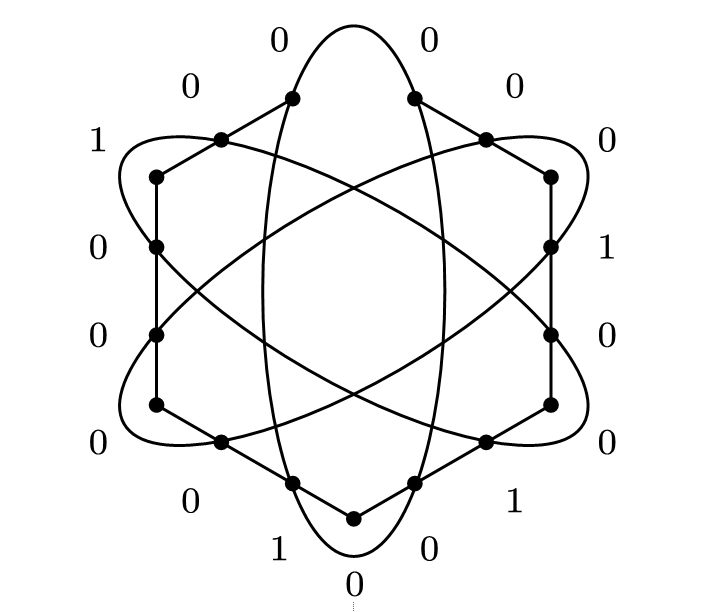}
    \caption{A KS assignment on the 17-vectors set}
    \label{CEG17}
\end{figure}

Note that the KS assignment in Fig. \ref{CEG17} violates the local consistency, because it forces the removed $(1,0,0,0)$ to be assigned 0 and 1 in two distinct contexts, so it is not a restriction of a state.

This example indicates that the nonexistence of KS assignments is strictly stronger than the KS contextuality. In other words, the KS contextual systems can be not generated by KS sets. Therefore, it is possible to find a set simpler than the CEG set to present KS contextuality, especially for the general projectors. The question of finding the minimal KS contextual system is valuable, because fewer projectors imply a simpler setup in experiments.

\subsection{State-independent contextuality}

Just as the KS contextuality does, the state-independent contextuality (SIC)\citep{Budroni2022Kochen} also characterizes the features of specific quantum systems. Nevertheless, SIC cannot be detached from the quantum states.

\begin{definition}
$Q\in QS$ is said to be \textbf{state-independently contextual (SI contextual)} if all $\rho\in qs(Q)$ are contextual.
\end{definition}

If we say that ``all $p\in s(Q)$ are contextual", it means that $s_{NC}(Q)$ is empty, so $Q$ is KS contextual. However, there exist quantum systems which are SI contextual but not KS contextual. A significant example is the Yu-Oh set\citep{Yu2012State}. It consists of 13 3-dimensional vectors.

\begin{figure}[H]
\centering
\includegraphics[width=0.5\linewidth]{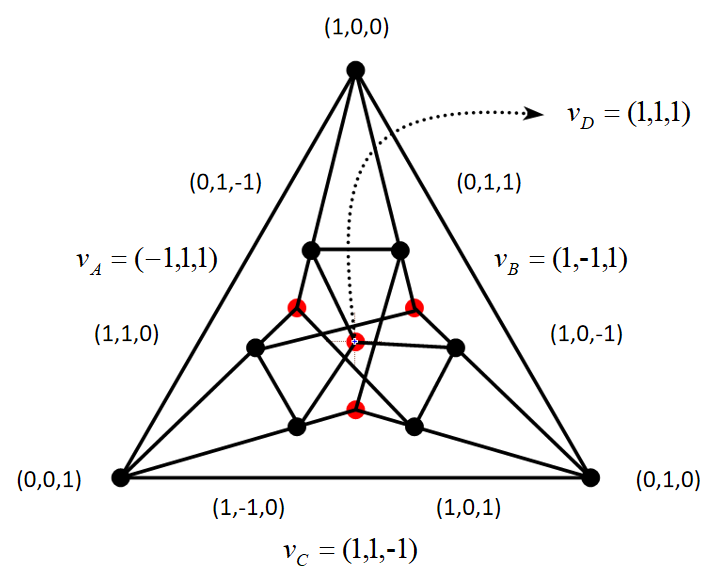}
\caption{The Yu-Oh set}
\label{Yu-Oh}
\end{figure}

 Let $Q_{Yu-Oh}$ be the quantum system generated by Yu-Oh set. It is easy to verify that $s_{01}(Q_{Yu-Oh})$ is not empty. However, Yu and Oh showed that $p(v_A)+p(v_B)+p(v_C)+p(v_D)\leq1$ for all $p\in s_{NC}(Q_{Yu-Oh})$, but  $\rho(v_A)+\rho(v_B)+\rho(v_C)+\rho(v_D)=\frac{4}{3}$ for all $\rho\in qs(Q_{Yu-Oh})$ because $|v_A\rangle\langle v_A|+|v_B\rangle\langle v_B|+|v_C\rangle\langle v_C|+|v_D\rangle\langle v_D|=\frac{4}{3}\hat{1}$. Therefore, $Q_{Yu-Oh}$ is SI contextual but not KS contextual.\par

Similar to the minimal KS set, the minimal rank-1 SIC set has been proved to contain at least 13 vectors\citep{Cabello2016Quantum}. Thus the Yu-Oh set is a minimal rank-1 SIC set. However, the minimal SI contextual system is still unknown\citep{Budroni2022Kochen}.

\section{The hierarchy of quantum contextuality}\label{sec-hierarchy}
The hierarchy of quantum contextuality is divided into two parts: the comparison between different quantum systems, and the the comparison between different quantum states on the same system. The hierarchy introduced by the sheaf theory approach\citep{Abramsky2011sheaf} only considers the latter, and  does not recognize the contextuality arguments such as the Yu-Oh set\citep{Budroni2022Kochen}.

If we use $s_{C}(Q)$, $s_{LC}(Q)$, $s_{SC}(Q)$, $s_{MC}(Q)$ and $s_{FC}(Q)$ to denote the contextual, logical contextual, strongly contextual, maximal contextual and fully contextual states on $Q\in QS$, the last section has shown that $s_{SC}(Q)\subseteq s_{LC}(Q)\subseteq s_{C}(Q)$, $s_{SC}(Q)\subseteq s_{MC}(Q)$ (equal for finite $Q$) and $s_{FC}(Q)\subseteq s_{MC}(Q)$. Therefore, the strong contextuality and the full contextuality are the strongest on a single quantum system.

In order to extend the contextuality hierarchy to capture the SIC argument, we prove that the state-independent strong contextuality is exactly the KS contextuality. Therefore, KS contextuality is the strongest when considering both the strength and proportion of contextual quantum states. If $Q\in QS$ is KS contextual, the theorem \ref{KS-strong} deduces that all $\rho\in qs(Q)$ are strongly contextual. For finite dimensional quantum systems, the converse is proved below.

\begin{theorem}
Let $Q$ be a quantum system on a finite dimensional Hilbert space $\mathcal{H}$. Then $Q$ is KS contextual if all $\rho\in qs(Q)$ are strongly contextual.\label{KS-SISC}
\end{theorem}
\begin{proof}
Suppose that $s_{01}(Q)$ is not empty, then there exists $\delta\in s_{01}(Q)$. Let $dim(\mathcal{H})=d$. The maximally mixed state $\hat{1}/d\in qs(Q)$ is strongly contextual, so there exists $\hat{P}\in Q$ such that $\delta(\hat{P})=1$ and $tr(\hat{P}\cdot\hat{1}/d)=0$. It is impossible because $tr(\hat{P}\cdot\hat{1}/d)>0$ for all $\hat{P}\neq\hat{0}$. Thus $s_{01}(Q)$ is empty.
\end{proof}

Therefore, a finite dimensional quantum system is KS contextual iff it is state-independently strong contextual, iff $\hat{1}/d$ is strongly contextual, which provides a prospective approach to determine KS contextual systems. We conjecture that theorem \ref{KS-SISC} also holds for general quantum systems.

The Fig. \ref{hierarchy} shows the hierarchy of quantum contextuality with two dimensions. The horizontal axis represents the proportion of contextual quantum states, while the vertical axis represents the strength of contextuality. C represents the contextuality, NC the non-contextuality, FC, SC, MC, LC, SDC and SIC the full, strong, maximal, logical, state-dependent and state-independent contextuality. The data point $Q_{Yu-Oh}$ etc. represent the specific quantum systems. $Q_{KS}$ represents the KS contextual quantum systems, such as the CEG system.

\begin{figure}[H]
\centering
\includegraphics[width=0.8\linewidth]{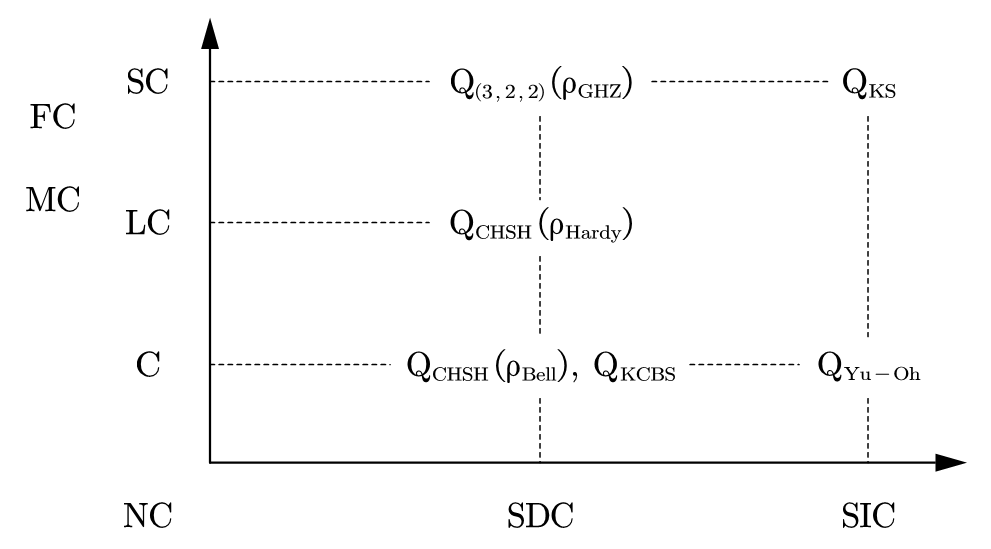}
\caption{The schematic illustration of quantum contextuality hierarchy (for finite dimensional quantum systems).}
\label{hierarchy}
\end{figure}

Some details in Fig. \ref{hierarchy} are still required further research, including the relation between FC and SC, or the relation between MC and LC. Furthermore, it has yet been verified whether $Q_{KCBS}$ or $Q_{Yu-Oh}$ contains logically contextual quantum states.

\section{Conclusion and future}\label{sec-conclusion}

We introduce two known observable-based contextuality theories, the sheaf theory approach and the exclusivity graph approach. The observable-based definitions and the neglect of local consistency or exclusivity principle have led to challenges faced by them. To overcome these challenges, we present an event-based contextuality theory with the $epBA$, which exactly describes the contextual systems with local consistency and logical exclusivity principle. Our theory formalizes the notion of contextuality (Theorem \ref{thm-contextuality}), precisely captures the structure of scenarios, and eliminates the ambiguity in the observable-based definitions.

By defining the quantum systems $QS\subseteq epBA$, we formalize a number of significant notions of quantum contextuality. The major quantum scenarios and models are precisely described in a unified mathematical framework, which allows us to extend the hierarchy of quantum contextuality to capture the SIC. Specifically, we prove that the KS contextuality is equivalent to the state-independent strong contextuality for finite dimensional quantum systems (Theorem \ref{KS-SISC}).

The $epBA$ enables the transformation of previously obscure contextuality issues into mathematical problems. We list some interesting questions for further research.

1. The minimal KS set is the CEG set\citep{Xu2020Proof}, but the minimal KS contextual system is still unknown. Note that the notion of KS set is not equivalent to the KS contextuality. In fact, the requirements for the KS sets are stronger.

2. The minimal rank-1 SIC set is the Yu-Oh set\citep{Cabello2016Quantum}, but the minimal SI contextual system is still unknown.

3. We have $s_{SC}(Q)\subseteq s_{MC}(Q)$ and $s_{FC}(Q)\subseteq s_{MC}(Q)$, while the relationship between SC and FC is not clear. We conjecture that $s_{SC}(Q)\subseteq s_{FC}(Q)$, which requires to find the NC inequalities that are attained maximum by any strongly contextual state, or to prove that the $\rho\in s_{SC}(Q)$ lie on the boundary of the polytope $s(Q)$. Similarly, the relationship between $MC$ and $LC$ is also unknown.

4. The comparison among the contextuality of different quantum systems remains a tough issue. Not only quantum states on a single quantum system witness different contextuality, but also it is not straightforward to assume that $\rho\in qs(Q)$ and $\rho'\in qs(Q')$ are equally strong, even though $\rho$ and $\rho'$ are both strongly contextual and so on. In fact, an ideal hierarchy may be linked to the advantages of quantum computation or quantum information processing, which still requires further investigation.

\bmhead{Acknowledgments}
The work was supported by National Natural Science Foundation of China (Grant No. 12371016, 11871083) and National Key R\&D Program of China (Grant No. 2020YFE0204200).

\bibliography{sn-bibliography}

\end{document}